\documentclass[11pt]{article}

\usepackage{fullpage}
\usepackage{amsmath,amsfonts,amsthm,mathrsfs,mathpazo,xspace,hyperref,graphicx}
\usepackage{endnotes}
\usepackage{color}
\usepackage{bm}
\usepackage{times}
\usepackage{amssymb,latexsym}

\newcommand{\ket}[1]{|#1\rangle}
\newcommand{\bra}[1]{\langle#1|}

\newcommand{\tensor}{\otimes}
\newcommand{\Tr}{\mbox{\rm Tr}}

\DeclareMathOperator{\poly}{poly}

\newcommand{\beq}{\begin{equation}}
\newcommand{\eeq}{\end{equation}}
\newcommand{\beqn}{\begin{equation*}}
\newcommand{\eeqn}{\end{equation*}}

\newcommand{\Es}[1]{\textsc{E}_{#1}}

\newcommand{\N}{\ensuremath{\mathbb{N}}}

\newcommand{\R}{\ensuremath{\mathbb{R}}}

\newcommand{\setft}[1]{\mathrm{#1}}

\newcommand{\density}[1]{\setft{D}\left(#1\right)}

\newtheorem{theorem}{Theorem}

\newtheorem{lemma}[theorem]{Lemma}
\newtheorem{claim}[theorem]{Claim}

\newtheorem{definition}[theorem]{Definition}

\newcommand{\be}{\begin{eqnarray}}
\newcommand{\ee}{\end{eqnarray}}

\newcommand{\Id}{\ensuremath{\mathop{\rm Id}\nolimits}}

\newcommand{\eps}{\varepsilon}

\definecolor{mygrey}{gray}{0.50}

\newcommand{\ignore}[1]{}

\newenvironment{step}
  {
    \begin{enumerate}

  }
  {\end{enumerate}}

\newenvironment{algorithm*}[1]
  {
    \begin{center}
      \hrulefill\\
      \textbf{#1}
  }
  {
    \vspace{-1\baselineskip}
    \hrulefill
    \end{center}
  }

\newenvironment{protocol*}[1]
  {
    \begin{center}
      \hrulefill\\
      \textbf{#1}
  }
  {
    \vspace{-1\baselineskip}
    \hrulefill
    \end{center}
  }

\newcommand{\CHSH}{{\textsc{CHSH}}}

\newcommand{\GUESS}{\textsc{GUESS}}

\newcommand{\VIOL}{\textsc{VIOL}}

\newcommand{\opt}{\texttt{opt}}

\newcommand{\boxA}{\mathcal{A}}
\newcommand{\boxB}{\mathcal{B}}
\newcommand{\boxE}{\mathcal{E}}
\newcommand{\Adv}{\textsc{Adv}}
\newcommand{\bl}{\mathbf{B}}
\newcommand{\ck}{\mathbf{C}}

\newcommand{\Hmin}{\textsc{H}_{min}}

\begin{document}

\title{Fully device independent quantum key distribution}
\author{Umesh Vazirani\thanks{Department of Computer Science, UC Berkeley, California. Supported by ARO Grant W911NF-09-1-0440 and NSF Grant CCF-0905626. Email \texttt{vazirani@eecs.berkeley.edu}}\qquad Thomas Vidick\thanks{Computer Science and Artificial Intelligence Laboratory, Massachusetts Institute of Technology. Supported by NSF Grant 0844626. Email \texttt{vidick@csail.mit.edu}.}}
\date{}
\maketitle

\begin{abstract}
The laws of quantum mechanics allow unconditionally secure key distribution protocols. Nevertheless, security proofs of traditional quantum key distribution (QKD) protocols rely on a crucial assumption, the trustworthiness of the quantum devices used in the protocol. In device-independent QKD, even this last assumption is relaxed: the devices used in the protocol may have been adversarially prepared, and there is no a priori guarantee that they perform according to specification. Proving security in this setting had been a central open problem in quantum cryptography. 

We give the first device-independent proof of security of a protocol for quantum key distribution that guarantees the extraction of a linear amount of key even when the devices are subject to a constant rate
of noise. Our only assumptions are that the laboratories in which each party holds his or her own device are spatially isolated, and that both devices, as well as the eavesdropper, are bound by the laws of  quantum mechanics. All previous proofs of security relied either on the use of many independent pairs of devices, or on the absence of noise. 
\end{abstract}

\section{Introduction}

Quantum key distribution~\cite{BB84,Eke91} together with its proof of security~\cite{Mayers01,SP00} appeared to have achieved the holy
grail of cryptography --- unconditional security, or a scheme whose security was based solely on the laws of physics. However, 
practical implementations of QKD protocols necessarily involve imperfect devices~\cite{Bennett92,Muller97}, and it was soon realized that these imperfections could be exploited by  a malicious eavesdropper to break the ``unconditional'' security of QKD (see e.g.~\cite{SK09} for a review).

Mayers and Yao~\cite{MY98} put forth a vision for restoring unconditional security in the presence of imperfect or even 
maliciously designed devices, by subjecting them to tests that they fail unless they behave consistently with ``honest'' devices. The fundamental challenge they introduced was of 
\emph{device-independent quantum key distribution} (DIQKD): establishing the security of a QKD protocol 
based only on the validity of quantum mechanics, the physical isolation of the devices and the passing of 
certain statistical tests. The germ of the idea for device-independence may already be seen in Ekert's original entanglement-based protocol for QKD~\cite{Eke91}, and was made more explicit by Barrett, Hardy, and Kent~\cite{BHK05}, who showed how to generate a single random bit secure against any non-signalling eavesdropper. A long line of research on DIQKD seeks to make the qualitative argument from~\cite{BHK05} quantitative, devising protocols that extract an amount of key that is linear in the number of uses of the devices, and is secure against increasingly general eavesdropping strategies. Initial works~\cite{AGM06,AMP06,Scarani06} give efficient and noise-tolerant protocols that are secure against individual attacks by non-signalling eavesdroppers. Subsequent work~\cite{MRCW09,Masanes09} and~\cite{HRW10} also proved security against collective attacks. Other works~\cite{ABGM07,PAB09,MRCW09,HR10,MPA11} obtain better key rates under the stronger assumption that the eavesdropper is bound by the laws of quantum mechanics. All these results, however, could only be established under restrictive \emph{independence} assumptions on the devices, e.g.
in recent work~\cite{HR10,MPA11} a proof of security based on collected statistics requires that 
the $n$ uses of each device are causally independent: measurements performed at successive steps of the protocol 
commute with each other. 

Very recently two papers~\cite{Barrett12,RUV12} announced proofs of security of DIQKD without requiring any
independence assumption between the different uses of the devices.
Unfortunately, although the approaches in~\cite{Barrett12,RUV12} are very different both implied protocols are polynomially inefficient and unable to tolerate noisy devices. The protocol used in~\cite{Barrett12} is very similar to the one originally introduced in~\cite{BHK05}, and requires a large number of uses of a pair of noise-free devices in order to generate a single bit of key. In the case of~\cite{RUV12}, DIQKD is obtained as a corollary of very strong testing that allows the shared quantum state and operators of the two untrusted devices to be completely characterized. It is an open question whether such strong testing can be achieved in a manner that is robust to noise. 

A major issue in QKD is dealing with the noise inherent in even the best devices. 
Indeed, a good DIQKD protocol should differentiate devices that
are ``honest but noisy" from devices that may attempt to take advantage of the protocol's necessary 
noise tolerance in order to leak information to an eavesdropper by introducing correlations in their ``errors''~\cite{BCK12}. The protocols in~\cite{Barrett12,RUV12} do not achieve this, since they cannot tolerate any constant noise rate.
This raises the question: is device-independent QKD even \emph{possible} without independence assumptions in a realistic, noise-tolerant scenario?

\subsection{Results}

We answer this question in the affirmative by giving the first complete device-independent proof of security of quantum key distribution that tolerates a constant noise rate and guarantees the generation of a linear amount of key. Our only assumption on the devices is that they can be modeled by the laws of quantum mechanics, and that they are spatially isolated from each other and from any adversary's laboratory. In particular, we emphasize that the devices may have quantum memory. While the proof of security is quite non-trivial (it builds upon ideas from the work on certifiable randomness generation mentioned below), the actual protocol whose device independence properties we establish is quite simple. It is a small variant of Ekert's entanglement-based protocol~\cite{Eke91}. 

In the protocol, the users Alice and Bob make $m$ successive uses of their respective devices. At each step, Alice (resp. Bob) privately chooses a random input $x_i\in\{0,1,2\}$ (resp. $y_i\in\{0,1\}$) for her device, collecting an output bit $a_i$ (resp. $b_i$). If the devices were honestly implemented they would share Bell states 
$\ket{\psi} = 1/\sqrt{2} \ket{00} + 1/\sqrt{2} \ket{11}$, and measure their qubits according to the following strategy: if $x_i = 0$ measure in the computational basis, if $x_i = 1$ measure in 
the Hadamard basis and if $x_i = 2$ measure in the $3\pi/8$-rotated basis. If $y_i = 0$ measure in the $\pi/8$-rotated basis and if $y_i = 1$ measure in the $3\pi/8$-rotated basis.

To test the devices, after the $m$ steps have been completed, the users select a random subset $\bl\subseteq \{1,\ldots,m\}$ of size $|\bl|=\gamma m$, where $\gamma>0$ is a small constant, and publicly announce their inputs and outputs in $\bl$. Rounds in $\bl$ will be called ``Bell rounds''. Let $z_i = 1$ if and only if $a_i\neq 2$ and $a_i\oplus b_i \neq x_i\wedge y_i$, or $(a_i,b_i)=(2,1)$ and $a_i\neq b_i$. The users jointly compute the noise rate $\eta := (1/|\bl|)\sum_{i\in \bl} z_i-(1-\opt)$, where $\opt = (2\cos^2\pi/8+1)/3$.\footnote{This corresponds to estimating the average amount by which the devices' outputs in $\bl$ differ from a maximal violation of a Bell inequality based on the CHSH inequality~\cite{Clauser:69a,BraunsteinC90}: see Section~\ref{sec:preliminaries} for details.} If $\eta \geq 0.5\%$, say, they abort. If not, they announce their remaining input choices. Let $\ck\subseteq\{1,\ldots,m\}$ be the steps in which $(a_i,b_i)=(2,1)$. We will call the rounds in $\ck$ the ``check rounds''; outputs from the rounds $\ck - \bl$ constitute the raw key. The users conclude by performing standard information reconciliation and privacy amplification steps, extracting a key of length $\kappa m$ for some $\kappa = \kappa(\eta,\eps)$, where $\eps$ is the desired security parameter. (We refer to Figures~\ref{fig:qkda} and~\ref{fig:qkdb} for a more detailed description of the protocol.)

\begin{theorem}[Informal]\label{thm:main-inf}
Let $m$ be a large enough integer and $\eps = 2^{-c_0 m}$, where $c_0>0$ is a small constant.
Given any pair of spatially isolated quantum devices $\mathcal{A}$ and $\mathcal{B}$, the protocol described above generates a shared key $K$ of length $\kappa m$, where $\kappa \approx 1.4\%$, that is $\eps$-secure: the probability that the users Alice and Bob do not abort and that the adversary can obtain information about the key is at most $\eps$. 
\end{theorem}

This informal statement hides a tradeoff between the parameters $\eps$, $\eta$, and $\kappa$: the larger the security parameter $\eps$ and the smaller the noise rate $\eta$, the higher the key rate $\kappa$. As $\eta\to 0$ (provided $\eps$ is chosen large enough) our proof guarantees a secure key rate $\kappa \approx 2.5\%$, which with our setting of parameters corresponds to about $15\%$ of the raw key. Conversely, the maximum noise rate for which we may extract a key of positive length is $\eta_{max}\approx 1.2\%$. This is worse than the optimal key rates obtained under the causal independence assumption~\cite{MPA11}, but still quite reasonable. 

\subsection{Proof overview and techniques}\label{sec:intro-tech}

We start with the observation that the randomness in the shared secret key must necessarily be generated by the two devices. 
Indeed, even though the users have the ability to generate perfect random bits privately, such bits cannot be used directly for the shared key, since any information transmitted about them is also available to the adversary. It follows that a necessary condition for DIQKD is that the users should be able to use their untrusted devices to generate \emph{certified} randomness --- randomness they can guarantee was not pre-encoded in the devices by the adversary, nor obtained as some function of the users' inputs to the devices.

Luckily, the possibility of generating certified randomness has already been investigated.  Building on an observation made in~\cite{Colbeck09}, Pironio et al.~\cite{Pironio} devised a protocol in which the generation of randomness could be certified solely by testing for a sufficiently large Bell inequality violation. In~\cite{FGS11,PM11} it was further shown that the randomness generated was secure against an arbitrary classical adversary. Concurrently, in~\cite{VV12} we gave a protocol that was secure even against a quantum adversary. This last protocol provides us with a solid starting point for DIQKD, since our goal is to prove that the quantum adversary, who may have fabricated the two devices, has no information about the shared random key. Nevertheless, extending this to DIQKD presents us with some serious new challenges. 

\begin{enumerate}
\item First, QKD is a task that involves two distant parties Alice and Bob. Any classical communication between Alice and Bob must take place in the clear and is therefore accessible to the adversary, thus giving her additional power. 

\item Second, in order to achieve QKD it is not sufficient just to generate randomness --- the point of QKD is that Alice and Bob share the same random key. In our protocol this is accomplished by distinguishing two different types of rounds: Bell rounds, in which the violation of the CHSH inequality by the devices is estimated, and check rounds, in which the devices are supposed to produce identical outputs from which the key will be generated. Unfortunately Alice and Bob must exchange information about which rounds are which, and since the adversary has access to all communicated classical information, this appears to render the Bell rounds pointless, since the adversary can ignore the Bell rounds and attack only those rounds which are used to generate the key (the check rounds). 

\item Finally, to be practical the protocol should tolerate noisy devices. As a result, the users can only expect a non-maximal amount of correlation, both in the Bell and check rounds. The randomness-certification protocol from~\cite{VV12} did not tolerate any noise --- in fact, the absence of noise played a crucial role in the proof. As we already explained in the introduction, dealing with the presence of noise is one of the major conceptual and technical hurdles of the proof. 
\end{enumerate}

We now explain how our proof technique addresses these challenges. The proof proceeds in two steps. As a first step, we argue that the following three conditions cannot hold simultaneously in any single round of the protocol: (i) the devices violate the CHSH inequality, whenever the round was selected as a Bell round (ii) the adversary can predict Bob's output, whenever the round was selected as a check round, and (iii) the no-signalling condition is satisfied between all three parties (Alice, Bob and the adversary). 
To derive a contradiction from (i)--(iii) we use a simple conceptual tool called the ``guessing game'', which was introduced in~\cite{VV12}. The main idea is that conditions (i) and (ii) imply that the adversary and Alice will be able to team up to predict Bob's output from their sole respective input/output behavior, violating the no-signalling condition (iii). 

\medskip

The second step is more challenging. All previous works on the subject reduced the general setting to a single-round scenario similar to the one outlined above by requiring some form of independence assumption on the devices or on the adversary's attack. We do not use any such assumption, and the main challenge is to deal with correlations between all rounds and the adversary in order to perform the reduction. 

Our starting point is the existence of a pair of devices that pass the protocol with non-negligible probability, but such that the adversary may gain non-negligible information about the secret key generated at the end of the protocol. Our goal is to show the existence of a round $i_0$ of the protocol in which conditions (i)--(iii) above are satisfied, thus deriving a contradiction. 

Our argument has two main ingredients. The first ingredient is the so-called ``quantum reconstruction paradigm'', a technique that was introduced in~\cite{DV09} and further developed in~\cite{DVPR11,VV12}. What this achieves is the following: any adversary able to obtain non-negligible information about the generated key can be transformed into a seemingly much stronger adversary: she can \emph{predict} the entire string of outputs of Bob's device on the check rounds (the rounds used to generate the key). Furthermore, the success probability of this ``guessing measurement'' is of the same order as the original distinguishing probability but does not depend on the length of the key --- a fact that will be crucial to obtaining good parameters. In order to achieve this, the new adversary  requires access to the same public information as the original one, together with a small number of additional ``advice bits'' taken from Bob's string of outputs.

This stronger form of the adversary guarantees that condition (ii) above holds in all rounds with small but non-negligible probability. Furthermore, the checking performed as part of the protocol ensures that (i) also holds on average over all rounds, with probability of the same order. 
The natural idea in order to identify a round $i_0$ in which conditions (i) and (ii) hold simultaneously with high probability is to perform conditioning: there must exist many rounds $i$ such that, provided both conditions hold in rounds $1$ to $i-1$, they must hold in round $i$ with high probability. 

Such conditioning, however, presents a new difficulty: it may introduce such correlations that condition (iii) is no longer satisfied. Indeed, recall that one of the main difficulties in analyzing the QKD protocol is that the adversary has considerable power, due to the large amount of public information that is leaked by the protocol --- including the users' complete choice of inputs. Hence conditioning on a low probability event involving the outcome of a measurement performed by the adversary on her system introduces correlations between inputs in all rounds. For instance, this conditioning could very well force the inputs in round $i_0$ to be a particular pair, say $(0,0)$, making the guarantees (i) and (ii) all but useless.  

The difficulty is reminiscent of one encountered in the analysis of parallel repetition, where conditioning on success in a subset of the parallel repeated games may introduce correlations among the players in the remaining games. Here, the situation is further complicated by the fact that it involves three parties involved in a relatively complex interaction. In particular, the conditioning is performed jointly on an event involving Alice and Bob (the CHSH violation observed in previous rounds being sufficiently large) on the one hand, and Bob and Eve (Eve's guess being correct) on the other. 

The final step in our proof consists in bounding the amount of correlation introduced by the conditioning. For this we use tools from information theory, including the chain rule for mutual information and the quantum Pinsker's inequality, which had not previously been applied to this setting. (Similar tools were already used by Holenstein in his derivation of a parallel repetition theorem for the case of two-player games with no-signalling players~\cite{Hol09}.)

\subsection{Perspective}

We have not attempted to optimize the relationship between the parameters $\kappa,\eta$ and $\eps$ describing the key rate, the noise rate and the security parameter respectively, and it is likely that the explicit dependency stated in Theorem~\ref{thm:main-tech} can be improved by tightening our arguments. It is an interesting question to find out whether our approach can lead to a trade-off as good as the one that has been shown to be achievable under additional assumptions on the devices~\cite{MPA11}. One possibility for improvement would be to bias the users' input distribution towards the pair of inputs $(2,1)$ from which the raw key is extracted, as was done in e.g.~\cite{AMP06}: indeed, only a very small fraction of the rounds are eventually required to estimate the violation of the $\CHSH$ condition.

Our proof crucially makes use of quantum mechanics to model the devices and the adversary. Can one obtain a fully device-independent proof of security of QKD against adversaries that are only restricted by the no-signalling principle? Barrett et al.~\cite{Barrett12} recently showed that such security is achievable in principle; however their protocol is highly inefficient and does not tolerate noisy devices. 



\paragraph{Organization of the paper.} We start with some preliminaries in Section~\ref{sec:preliminaries}, introducing our notation, the information-theoretic quantities that will be used. We also summarize the main parameters of our protocol, which is described in Figures~\ref{fig:qkda} and~\ref{fig:qkdb}. In Section~\ref{sec:analysis} we formally state our result and outline the security proof. The two main ingredients are the analysis of Protocol~B, which is given in Section~\ref{sec:chain}, and the ``quantum reconstruction paradigm'' introduced in Section~\ref{sec:adv}. Finally, Section~\ref{sec:additional} contains probabilistic and information-theoretic lemmas used in some of the proofs. 

\paragraph{Acknowledgments.} We thank Anthony Leverrier for many useful comments on a preliminary version of this manuscript.

\section{Preliminaries}\label{sec:preliminaries}

We assume familiarity with basic concepts and standard notation in quantum information, including density matrices and distance measures such as the trace distance and the fidelity. We refer the reader to the books~\cite{NC00,Wilde11} for detailed introductions.

\paragraph{Notation.} We use roman capitals $A,B,\ldots,X$ both to refer to random variables and the registers, classical or quantum, that contain them. Calligraphic letters $\mathcal{A},\mathcal{B},\ldots,\mathcal{X}$ are used to refer to the underlying Hilbert space. $\density{\mathcal{X}}$ denotes the set of density operators (non-negative matrices with trace $1$) on $\mathcal{X}$. For an arbitrary matrix $A$ on $\mathcal{X}$ we let $\|A\|_1 = \Tr\sqrt{AA^\dagger}$ denote its Schatten $1$-norm. $\ln$ denotes the natural logarithm and $\log$ the logarithm in base $2$. For $x\in [0,1]$, $H(x)=-x\log x-(1-x)\log(1-x)$ is the binary entropy function. 

\paragraph{Information theoretic quantities.}
Given a density matrix $\rho\in\density{\mathcal{A}}$, its von Neuman entropy is $H(\rho) := -\Tr(\rho\ln\rho)$. For a classical-quantum state $\rho_{XA} = \sum_x p_x\ket{x}\bra{x} \otimes \rho_x\in\density{\mathcal{X}\otimes\mathcal{A}}$, where for every $x$, $\rho_x\in\density{\mathcal{A}}$, the conditional entropy is defined as $H(A|X)_\rho:= \sum_x p_x H(\rho_x)$. Given a state $\rho_{ABX}$, where $X$ is classical, the conditional mutual information is  
$$ I(A:B|X)_\rho \,:=\, H(A|X)_\rho+H(B|X)_\rho-H(AB|X)_\rho.$$
We will use the following quantum analogue of the classical Pinsker's inequality (see e.g. Theorem~11.9.1 in~\cite{Wilde11} for a proof): for any $\rho_{AB}\in\density{\mathcal{AB}}$,
\beq\label{eq:pinsker}
\big\|\rho_{A B} - \rho_{A} \otimes \rho_{B} \big\|_1^2 \,\leq\, (2\ln 2)\, I(A:B)_{\rho}.
\eeq
The most important information measure in our context is the quantum conditional min-entropy, first introduced in~\cite{Ren05}, and defined as follows. 

\begin{definition}  \label{def:min-entropy}
  Let $\rho_{A B}$ be a bipartite density matrix. The
  \emph{min-entropy} of $A$ conditioned on $B$ is defined as
  \begin{align*}
    \Hmin(A|B)_\rho \,:=\, \max \{\lambda
    \in \R :  \exists \sigma_B \in \density{\mathcal{B}} 
		\,\mathrm{s.t.}\,\, 2^{-\lambda} \Id_A \tensor \sigma_B \geq \rho_{AB}\}.
  \end{align*}
\end{definition}

We will often drop the subscript $\rho$ when there is no doubt about
the underlying state. The smooth min-entropy is defined as follows. 

\begin{definition}  \label{def:smooth-min-entropy}
  Let $\eps \geq 0$ and $\rho_{AB}$ a bipartite density matrix. The
  \emph{$\eps$-smooth min-entropy} of $A$ conditioned on $B$ is defined as
  \begin{equation*}
    \Hmin^\eps(A|B)_\rho \,:=\, \max_{\tilde{\rho}_{AB}
      \in B(\rho_{AB},\eps)} \Hmin(A|B)_{\tilde{\rho}},
  \end{equation*}
  where $B(\rho_{AB},\eps)$ is a ball of
  radius $\eps$ around
  $\rho_{AB}$.\footnote{Theoretically any distance measure could be
    used to define an $\eps$-ball. As has become customary, we use the \emph{purified
      distance}, $P(\rho,\sigma) := \sqrt{1 -
      F(\rho,\sigma)^2}$, where $F(\cdot,\cdot)$ is the fidelity. }
\end{definition}

\paragraph{The $\CHSH$ condition.}
The security of our DIQKD protocol is based on the statistical verification that the pair of devices used have an input/output behavior consistent with certain pre-determined correlations, which are those expected of a ``honest'' quantum-mechanical pair of devices performing the measurements described below.  

Let $\boxA$ and $\boxB$ designate two spatially isolated devices. In the protocol, there are three possible choices of inputs $x\in\{0,1,2\}$ to $\boxA$, and two possible inputs $y\in\{0,1\}$ to $\boxB$. Each of the $6$ possible pairs of inputs is chosen with uniform probability $1/6$. The devices are required to produce outputs $a,b\in\{0,1\}$ respectively. 
The users select a random subset of the rounds of the protocol in which to evaluate the frequency with which the following constraints are satisfied. In case both inputs were in $\{0,1\}$, the constraint on the outputs is the CHSH parity constraint $a\oplus b = x\wedge y$~\cite{Clauser:69a}. If the inputs are $(2,1)$ the constraint is that the outputs $(a,b)$ should satisfy $a\oplus b=0$. Finally, for the remaining pair of inputs $(2,0)$ all pairs of outputs are valid. We will refer to this set of constraints collectively as ``the $\CHSH$ condition''. We note that the underlying Bell inequality is similar to the so-called ``chained inequality'' for two inputs~\cite{BraunsteinC90}. 

Let $\opt$ be the maximum probability with which any two isolated devices, obeying the laws of quantum mechanics, may produce outputs satisfying the $\CHSH$ condition. It is not hard to show that $\opt = (2/3)\cos^2\pi/8 + (1/3)$, which is achieved using the following strategy. The devices are initialized in a single EPR pair $\ket{\Psi}=(\ket{00}+\ket{11})/\sqrt{2}$, each device holding one qubit. On input $0$, $\boxA$ performs a measurement in the computational basis, and on input $1$ it measures in the Hadamard basis. On input $0$, $\boxB$ measures in the computational basis rotated by $\pi/8$. If $\boxA$ gets input $2$, or if $\boxB$ gets input $1$, they measure in the computational basis rotated by $3\pi/8$. The devices may be used repeatedly, and honest devices perform measurements on a fresh EPR pair at each use.

\paragraph{Parameters.}
For convenience, we summarize here the main parameters of the key distribution protocol described in Figures~\ref{fig:qkda} and~\ref{fig:qkdb}.
\begin{itemize}
\item $m$ is the total number of rounds in the protocol (in each round, an input to each of $\boxA,\boxB$ is chosen, and an output is collected).  
\item $\bl$ are the ``Bell rounds'', selected to perform parameter estimation. They are chosen uniformly at random under the constraint that $|\bl|=\gamma m$, for some $\gamma>0$ specified in the protocol.
\item $\eta$ is the tolerated error rate: the protocol aborts as soon as the fraction of rounds in $\bl$ satisfying the $\CHSH$ condition is lower than $\opt-\eta$.  
\item $\ck\subseteq [m]$ are the ``check rounds''. Those are rounds in which the inputs to $(\boxA,\boxB)$ are $(2,1)$. Since the inputs are chosen uniformly at random, the number of check rounds $|\ck|$ is highly concentrated around $m/6$.
\item The target min-entropy rate $\kappa$. This is the rate of min-entropy that the users Alice and Bob expect to be present in the check rounds, provided the protocol did not abort. Once information reconciliation and privacy amplification have been performed, a secret key of length roughly $(\kappa - H(2\eta))|\ck|$ will be produced. 
\item $\eps$ is the security parameter: the statistical distance from uniform of the extracted key (conditioned on the eavesdropper's side information). Precisely, if $K$ denotes the system containing the extracted key, we will obtain that $\|\rho_{K\mathcal{E}'} - \rho_{U_{K}}\otimes \rho_{\mathcal{E}'} \|_1 \leq \eps$, where $\mathcal{E}'$ is a register containing all the side information available to an arbitrary quantum eavesdropper in the protocol,  and $\rho_{U_{K}}$ is the totally mixed state on as qubits as the key length. 
\end{itemize}

\section{Analysis of the key distribution protocol}\label{sec:analysis}

\begin{figure}
\begin{protocol*}{Protocol~A}
\begin{step}
\item Let $m$ and $\eps,\eta>0$ be parameters given as input. Let $C_\gamma$ be the constant from Theorem~\ref{thm:main-tech}, and set $\gamma = (C_\gamma/\eta^2)\ln(1/\eps)/m$. 
\item\label{step:protb} Alice and Bob run Protocol~B for $m$ steps, choosing inputs $x \in \{0,1,2\}^m$ (resp. $y\in\{0,1\}^m$) and obtaining outcomes $a\in\{0,1\}^m$ (resp. $b\in\{0,1\}^m$). Let $\bl$ be the set of rounds that were chosen to perform parameter estimation. 
\item Alice and Bob publicly reveal their choices of inputs. Let $\ck$ be the set of rounds $i$ in which $(x_i,y_i)=(2,1)$. If $||\ck|- m/6|>10\sqrt{m}$ they abort the protocol.
\item\label{step:ir} Alice and Bob perform information reconciliation on their outputs in $\ck-\bl$, which constitute the raw key. For this, Bob sends a message of $\ell \leq H(2\eta)|\ck| + \log(2/\eps)$ bits to Alice. 
\item\label{step:pa} Let $\kappa=\kappa(\eta)$ be as specified in Theorem~\ref{thm:main-tech}. Alice and Bob perform privacy amplification using e.g. two-universal hashing, extracting a shared key of length $(\kappa-H(2\eta)-O(\log(1/\eps)/m)) |\ck|$ from the common $(|\ck|-|\bl|)$-bit string they obtained at the end of the previous step.
\end{step}
\end{protocol*}
\caption{The device-independent key distribution protocol, Protocol~A}
\label{fig:qkda}
\end{figure}

\begin{figure}
\begin{protocol*}{Protocol~B}
\begin{step}
\item Let $m,\gamma$ and $\eta$ be parameters given as input. 
\item Repeat, for $i=1,\ldots,m$:
\begin{step}
\item Alice picks $x_i\in\{0,1,2\}$, and Bob picks $y_i \in \{0,1\}$, uniformly at random. They input $x_i,y_i$ into their respective device, obtaining outputs $a_i,b_i\in\{0,1\}$ respectively.
\end{step}
\item\label{step:check} Alice chooses a random subset $\bl\subseteq [m]$ of size $\gamma m$ and shares it publicly with Bob. Alice and Bob announce their input/output pairs in $\bl$, and compute the fraction of pairs satisfying the $\CHSH$ condition. Let $(\opt-\eta')$ be this fraction. If $\eta'>\eta$ they abort the protocol. 
\end{step}
\end{protocol*}
\caption{Theorem~\ref{thm:main-tech} shows that, at the end of protocol~B, the bits $B_\ck$ generated by Bob's device in the check rounds $\ck$ both have high smooth min-entropy, conditioned on the adversary's arbitrary quantum side information.}
\label{fig:qkdb}
\end{figure}

The analysis of Protocol~A, and the proof of Theorem~\ref{thm:main-inf}, is performed in two steps. The first, main step consists in proving a lower bound on the quantum smooth  conditional min-entropy $H_{min}^\eps(B_{\ck}|XYA_\bl B_\bl\mathcal{E})$ of the outputs obtained by Bob in the check rounds $\ck$ (conditioned on the protocol not aborting). This lower bound will depend on the maximal error rate $\eta$ that is tolerated by the users in the sub-protocol~B (see Figures~\ref{fig:qkda} and~\ref{fig:qkdb} for a description of protocols~A and~B respectively). Here the lower bound is taken conditioned on the state of an arbitrary quantum adversary (whom we will call Eve and refer to indiscriminately as ``the adversary'' or ``the eavesdropper'') in the protocol, who has access to the information $X,Y,A_\bl,B_\bl$ revealed publicly in the course of the protocol, as well as to a quantum system $\mathcal{E}$ which may be correlated with the systems $\boxA$, $\boxB$ of the devices. Such an estimate is stated in Theorem~\ref{thm:main-tech} in Section~\ref{sec:main-proof} below.

The second step consists in showing that there exists appropriate protocols for the information reconciliation and privacy amplification steps, Steps~\ref{step:ir} and~\ref{step:pa} in Protocol~A respectively, such that the lower bound on the conditional min-entropy from the first step guarantees the security (distance from uniform from the point of view of the adversary) and correctness (Alice and Bob should obtain the same key) of the key that is extracted. This step is standard, and all the ingredients required already appear in the literature. We summarize the result as Lemma~\ref{lem:prota} in Section~\ref{sec:irpa} below. 

Theorem~\ref{thm:main-inf} follows immediately by combining Theorem~\ref{thm:main-tech} and Lemma~\ref{lem:prota}.

\subsection{Probability space}\label{sec:probspace}

Before stating and proving formally our results, we formally define the random variables and events that will be used in their proof. 

\paragraph{Modeling the devices.}
Fix a pair of spatially isolated devices $(\boxA,\boxB)$. Device $\boxA$ takes inputs in $\{0,1,2\}$, and device $\boxB$ takes inputs in $\{0,1\}$. Whenever provided an input, each device produces an output in $\{0,1\}$. The devices may be used repeatedly. We will assume that the pair $(\boxA,\boxB)$ can be described by quantum mechanics: the devices are modeled by a pair of quantum registers; when provided an input each device performs a measurement on the state contained in the corresponding subsystem.

We assume that user Alice holds $\boxA$, and Bob is given $\boxB$. In addition, there is an adversary Eve who holds an additional quantum register $\mathcal{E}$, initialized in a state arbitrarily correlated with that of $\mathcal{A}$ and $\mathcal{B}$. Let $\rho_{A_1B_1\mathcal{E}}$ be the density matrix describing the joint state of all three registers at the start of the protocol. 

We define the following random variables and events. $X \in \{0,1,2\}^m$ and $Y\in\{0,1\}^m$ are two uniformly distributed random variables, used to represent the inputs to $\boxA,\boxB$ respectively, as chosen in the protocol. $A,B\in\{0,1\}^m$ are random variables denoting the outputs produced by the devices, when sequentially provided their respective inputs $X,Y$. We will always use $\ck\subseteq [m]$ to denote the set of ``check'' rounds, in which $(X_i,Y_i)=(2,1)$, and $\bl\subseteq [m]$ the set of ``Bell'' rounds chosen by Alice and Bob to perform parameter estimation. 

Let $\rho_{\boxA_i\boxB_i}$ denote the reduced state of devices $\boxA$ and $\boxB$ in the $i$-th round of the protocol (before they have been provided their $i$-th input). Formally,  
$$ \rho_{\boxA_i\boxB_i} \,\propto\, \Big( \prod_{j<i} M_{X_j}^{A_j} \otimes N_{Y_j}^{B_j} \Big)\,\rho_{\boxA_1\boxB_1}\,\Big( \prod_{j<i} \big(M_{X_j}^{A_j}\big)^\dagger \otimes \big(N_{Y_j}^{B_j}\big)^\dagger \Big),$$
where $\{M_{X_j}^{A_j}\}$ and $\{N_{Y_j}^{B_j}\}$ are the Kraus operators corresponding to the measurement performed by devices $\boxA$ and $\boxB$ in round $j$ respectively, and $\rho_{\boxA_i\boxB_i}$ is normalized. Here $\rho_{\boxA_1\boxB_1} = \Tr_\mathcal{E}(\rho_{\boxA_1\boxB_1\mathcal{E}})$ is the reduced state of the devices at the start of the protocol. It is important to note that for any $i$ the state $\rho_{\boxA_i\boxB_i}$ may depend on a measurement that is performed on system $\mathcal{E}$ as soon as a particular outcome of that measurement is fixed. 

\paragraph{Measuring the $\CHSH$ condition.}
Given a set $S\subseteq [m]$ and $\delta>0$, $\CHSH_{\boxA\boxB}(S,\delta)$ is the event that the tuple $(X,Y,A,B)$ satisfies the $\CHSH$ condition (as described in Section~\ref{sec:preliminaries}) in a fraction at least $\opt-\delta$ of the rounds indicated by $S$. If $S$ is omitted, $\CHSH_{\boxA\boxB}(\delta) = \CHSH_{\boxA\boxB}([m],\delta)$. 
Letting $Z \in \{0,1\}^m$ be the indicator random variable of the CHSH condition \emph{not} being satisfied in any given round, we can write
$$\CHSH_{\boxA\boxB}(S,\delta) \,\equiv\, \Big\{\frac{1}{|S|}\sum_{i\in S} Z_i \leq (1-\opt)+\delta\Big\}.$$
We also define $\VIOL_{\boxA\boxB}(i)$, where $i\in [m]$, to express the expected amount by which the $\CHSH$ condition in round $i$ is satisfied: 
$$\VIOL_{\boxA\boxB}(i)\,=\, \textrm{E}[\,Z_i\,] - (1-\opt),$$
where here the expectation is taken over the choice of inputs $(X_i,Y_i)$ in round $i$, and over the randomness in the devices' own measurements in round $i$. Note that $\VIOL_{\boxA\boxB}(i)$ implicitly depends on the specific state of the devices in round $i$, which may be affected by previous input and outputs obtained in the protocol as well as on other events that may be conditioned on. Hence the expression $\Pr( \VIOL_{\boxA\boxB}(i) < \delta | E)$, for some event $E$, indicates the average probability, over all possible $e\in E$, that the devices satisfy the $\CHSH$ condition in round $i$ with probability at least $\opt-\delta$, provided their inputs are distributed according to the conditional distribution $(X_i,Y_i)|E=e$, and when performed on the post-measurement state of $\boxA\otimes \boxB$ in round $i$ conditioned on $E=e$.
For any $\delta>0$ we let $\VIOL_{\boxA\boxB}(\delta)$ be the event that $(1/m)\sum_i \VIOL_{\boxA\boxB}(i) \leq \delta$.

\paragraph{The adversary.}
We introduce additional random variables that depend on the adversary Eve, holding the quantum register $\mathcal{E}$. The adversary is described in Lemma~\ref{lem:strong-adv} below; to understand the events below it may be useful to read that lemma's statement first. 

Let $E \in \{0,1\}^{|\ck|}$ be the random variable that describes the outcome of the measurement on $\mathcal{E}$ described in Lemma~\ref{lem:strong-adv}. Note that this outcome depends on the ``advice'' that is given to the adversary. We use $\hat{X},\hat{Y}$ to denote the inputs that are given to the adversary, and $\hat{\Adv} \in\{0,1\}^{\alpha m}$ to denote the additional advice bits. These random variables need not equal the actual values $X,Y,\Adv$: in general, the adversary's measurement is well-defined for any given advice bits, and $E$ is used to denote its outcome irrespective of whether the advice given was ``correct'' or not. 
For any $i\in [m]$, define $\GUESS_{\boxB\boxE}(i)\in\{0,1\}$ to be $1$ if and only if, either $i\in \ck$ and $E_i = B_i$, or $i\notin \ck$, and let $\GUESS_{\boxB\boxE} = \wedge_i \GUESS_{\boxB\boxE}(i)$.

\subsection{Information reconciliation and privacy amplification}\label{sec:irpa}

For convenience, we let $\boxE' := XYA_\bl B_\bl \boxE$ denote the side information available to the eavesdropper. We show the following lemma, whose proof follows from standard arguments in the analysis of QKD protocols (see e.g.~\cite{Ren05}). We provide the relevant details below. 

\begin{lemma}\label{lem:prota}
Let $\gamma,\eps>0$. Let $\eps' = 2e^{-\gamma|\ck|/400}$. Suppose that, after Step~\ref{step:protb} of Protocol~A, the condition $\Hmin^\eps(B_\ck|\mathcal{E}') \geq \kappa |\ck|$ is satisfied. Then with probability at least $1-\eps'$, at the end of the protocol Alice and Bob have a common shared key that is $2\eps$-close to uniform and has length $\Hmin^\eps(B_\ck|\mathcal{E}') - H(1.1\eta)|\ck| - 4\log(1/\eps)$.
\end{lemma}

\paragraph{Information reconciliation.}
We first analyze the information reconciliation step. The following lemma states the conditions that are required for there to exist a satisfactory information reconciliation procedure. 

\begin{lemma}[Lemma~6.3.4 in~\cite{Ren05}]\label{lem:ir}
Let $A,B\in \{0,1\}^k$ be two random variables, and $\eps>0$. Suppose Alice holds $A$, and Bob holds $B$. There is an information reconciliation protocol in which Bob communicates $\ell \leq H_{max}^\eps(B|A) + \log(2/\eps)$ bits of information about $B$ to Alice and is such that with probability at least $1-\eps$ Alice and Bob both know $B$ at the end of the protocol. 
\end{lemma}

To apply Lemma~\ref{lem:ir} it suffices to prove an upper bound on the conditional max-entropy $H_{max}^\eps(B_{\ck}|A_{\ck})$. By definition of the rounds $\ck$, the $\CHSH$ condition in those rounds imposes that $A_i = B_i$ for all $i\in \ck$. Hence, were it not for errors, we would have $H_{\max}^\eps(B|A) = 0$. The following claim shows that the bound on the error rate that results from the estimation performed in the rounds $\bl$ in Step~\ref{step:check} of Protocol~B is enough to guarantee a good upper bound on the conditional max-entropy. 

\begin{claim}\label{claim:irbound} Suppose Alice and Bob do not abort after Step~\ref{step:check} in Protocol~B. Let $\ck$ be the set of check rounds, as designated in Step~\ref{step:ir} of Protocol~A. Then 
$H_{max}^{\eps'}(B_\ck|C_\ck) \,\leq\, H(1.1\eta)|\ck|$, where $\eps' = 2e^{-\gamma|\ck|/400}$. 
\end{claim}

\begin{proof}
Fix the set $\ck$. The set $\bl$ chosen by Alice and Bob to perform parameter estimation contains a fraction at least $\gamma/2$ of the rounds in $\ck$, except with probability at most $e^{-\gamma|\ck|/8}$. The protocol is aborted as soon as more than an $\eta$ fraction of those rounds are such that $a_i\neq b_i$. Hence with probability at least $1-e^{-\gamma |\ck|/200}$ the total fraction of errors in $\ck$ is at most $1.1\eta$. In particular, with probability at least $1-e^{-\gamma|\ck|/400}$ over $A_{\ck}$, with probability at least $1-e^{-\gamma|\ck|/400}$, $B_{\ck}$ will take on at most $2^{H(1.1\eta)|\ck|}$ values. 
\end{proof}

\paragraph{Privacy amplification.}
The following lemma states the existence of a good protocol for privacy amplification. 

\begin{lemma}[Lemma 6.4.1 in~\cite{Ren05}]\label{lem:pa}
Suppose the information reconciliation protocol requires at most $\ell$ bits of communication. Then for any $\eps>0$ there is a privacy amplification protocol based on two-universal hashing which extracts $\Hmin^\eps(B_{\ck}|\mathcal{E}') - \ell - 2\log(1/\eps)$ bits of key. 
\end{lemma}

Lemma~\ref{lem:prota} now follows directly by combining Claim~\ref{claim:irbound} with Lemma~\ref{lem:pa} and the assumption on the conditional min-entropy placed in the lemma.

\subsection{A lower bound on the conditional min-entropy}\label{sec:main-proof}

The main result of this section is a lower bound on the conditional smooth min-entropy $\Hmin^\eps(B_{\ck}|XYA_\bl B_\bl \boxE)$ of the raw key.

\begin{theorem}\label{thm:main-tech}
Let $\eta >0$ be given. There exists positive constants $C_\eps,C_\gamma$ (possibly depending on $\eta$) such that the following hold. Let $m$ be an integer and $\eps\geq e^{-C_\eps m}$ be given. Let $\gamma = (C_\gamma/\eta^2) \ln(1/\eps)/m$ be as specified in Protocol~A (Figure~\ref{fig:qkda}). Let $\kappa$ be any constant such that $\kappa< (\sqrt{2}-1)/(4\ln(2)) - (4/\ln(2))\eta$.

Suppose that the devices $\boxA$, $\boxB$ are such that with probability at least $\eps$ the protocol does not abort. Let $\mathcal{E}$ be an auxiliary system held by an eavesdropper, who may also learn $(X,Y)$ and $(A_\bl,B_\bl)$. Then, conditioned on the protocol not aborting, it holds that 
$$\Hmin^\eps(B_{\ck}|XYA_\bl B_\bl \boxE) \geq \kappa |\ck| - O\big(\ln(1/\eps)\big).$$
\end{theorem}

We note that the precise relation between the parameters $\kappa$ and $\eta$ stated in the theorem is the one that we obtain from our proof; however we have not attempted to optimize it fully and it is likely that one may be able to derive a better dependency. It is also clear from the proof that one may trade off the different constants between each other, depending on whether one is interested in the maximum possible key rate in the presence of very small noise, or to the opposite if one wishes to tolerate as much noise as possible. 

\medskip

The proof of Theorem~\ref{thm:main-tech} is based on three lemmas. We state the lemmas first, and derive the theorem from them below. 

\subsubsection{The reconstruction lemma}

Our first lemma states that, if the min-entropy condition in the conclusion of the theorem is not satisfied, then there must exist a measurement on the system $\boxE$, depending on $X,Y,A_\bl$ and $B_\bl$, together with some additional ``advice'' bits of information about $B_\ck$, whose outcome $E\in\{0,1\}^{|\ck|}$ agrees with $B_\ck$ with non-negligible probability.   

\begin{lemma}\label{lem:strong-adv}
Let $\kappa>0$ and suppose that $\Hmin^\eps(B_{\ck}|XYA_\bl B_\bl \boxE) < \kappa  |\ck|$. Then there exists an $\alpha = \kappa |\ck|/m +2\gamma  + O(\log (m/\eps)/m)$ and a function $f:\{0,1\}^{|\ck|}\to\{0,1\}^{(\alpha-2\gamma) m}$ such that, given the  bits $\Adv=f_\Adv(B_\ck)A_\bl B_\bl\in\{0,1\}^{\alpha m}$ together with the inputs $X,Y$, there exists a measurement on $\boxE$ that outputs a string $e\in\{0,1\}^{|\ck|}$ such that with probability (over the randomness in $B$ and in the measurement) at least $C_E(\eps/m)^6$, where $C_E$ is a universal constant, the equality $e= b_\ck$ holds.
\end{lemma}

The proof of Lemma~\ref{lem:strong-adv} is based on a ``reconstruction''-type argument from~\cite{DVPR11}. A very similar argument was already used to establish an analogous lemma in~\cite{VV12}. We give the proof of Lemma~\ref{lem:strong-adv} in Section~\ref{sec:adv}. 

\subsubsection{Existence of a good round}

Our second lemma states the existence of a ``good'' round $i_0\in [m]$ in which both the $\CHSH$ condition is satisfied, and the outcome $E_{i_0}$ of the measurement described in Lemma~\ref{lem:strong-adv} agrees with $B_{i_0}$, with good probability. Note also the additional condition~\eqref{eq:chain-0} in the lemma, which states that systems $\boxA$ and $\boxB$ are each close to being independent from the random variables $X_{i_0},Y_{i_0}$ describing the choice of inputs in round $i_0$. This condition is necessary for condition~\eqref{eq:chain-1}, on the $\CHSH$ violation, to be of any use: indeed, without~\eqref{eq:chain-0} it could in principle be that the conditioning on specific outcomes in previous rounds, including the adversary's outcomes, completely fixes the choice of inputs in the $i_0$-th round. Conditions~\eqref{eq:chain-0}--\eqref{eq:chain-2} in the lemma correspond to conditions~(i)--(iii) discussed in Section~\ref{sec:intro-tech}.  

Eq.~\eqref{eq:chain-0} implies that the distribution that arises from the devices' measurements on the states $\rho_{\boxA_{i_0}\boxB_{i_0}}$ is, while not necessarily quantum, still no-signalling, and this is all that is required for the application of the guessing lemma, Lemma~\ref{lem:guessing} below. As explained in the introduction, proving this condition is an important point of departure of our proof from previous approaches, which used an assumption of independence between the devices or a limitation of the adversary in order to automatically obtain that (an even stronger form of) the condition held in all rounds without requiring any conditioning.  

We refer to Section~\ref{sec:probspace} for a description of the events $\CHSH_{\boxA\boxB}$ and $\VIOL_{\boxA\boxB}$ appearing in the statement of the lemma.

\begin{lemma}\label{lem:chain-rule}
Let $\hat{\Adv}$ be uniformly distributed in $\{0,1\}^{\alpha m}$, and $\eta,\eps>0$ be such that the following holds:
$$\Pr\big( \CHSH_{\boxA\boxB}(\eta) \wedge \GUESS_{\boxB\boxE} | \Adv = \hat{\Adv} \big) \,\geq\, \eps,$$
and let $\alpha = |\Adv|/ m$. 
Then there exists a universal constant $C_\nu>0$, a $\nu \leq C_\nu\sqrt{\log(1/\eps)/m}$, an $i_0 \in [m]$ and a set $G_{i_0}\subseteq (\{0,1,2\}\times\{0,1\}\times \{0,1\}^3)^{i_0-1}$ such that for every $(x,y,a,b,e)\in G_{i_0}$, there is a choice of $\hat{x}_{> i_0},\hat{y}_{> i_0}$ and an $\hat{\Adv}$ consistent with $((x,\hat{x}_{> i_0}),(y,\hat{y}_{> i_0}),a,b)$ such that 
the following hold:
\begin{align}
&\max\Big\{  \Big\| \rho_{\boxA_{i_0} X_{i_0}Y_{i_0}} - \rho_{\boxA_{i_0}}\otimes \Big(\frac{1}{6}\sum_{x,y} \ket{x,y}\bra{x,y}\Big) \Big\|_1\ ,\notag\\
&\qquad\quad \Big\| \rho_{\boxB_{i_0} X_{i_0}Y_{i_0}} - \rho_{\boxB_{i_0}}\otimes \Big(\frac{1}{6}\sum_{x,y} \ket{x,y}\bra{x,y}\Big) \Big\|_1\Big\} \leq \nu, \label{eq:chain-0}\\
&\VIOL_{\boxA\boxB}(i_0) \,\leq\, 3\eta+\nu,\label{eq:chain-1}\\
&\Pr(\GUESS_{\boxB\boxE}(i_0))\geq  1-12\ln(2)\alpha - \nu,\label{eq:chain-2}
\end{align}
where in~\eqref{eq:chain-0} the state $\rho_{\boxA_{i_0}\boxB_{i_0} X_{i_0}Y_{i_0}}$ is the (normalized) state of the corresponding systems in round $i_0$, conditioned on $(x,y,a,b,e)$, and similarly in~\eqref{eq:chain-1} and~\eqref{eq:chain-2} the violation is estimated conditioned on previous input/outputs to the devices being $(x,y,a,b)$, and on Eve making her measurement based on the inputs $(x,2,\hat{x}_{> i_0})$ and $(y,1,\hat{y}_{> i_0})$ and advice string $\hat{\Adv}$, and obtaining outcomes $e$ as her prediction in rounds $\ck\cap \{1,\ldots,i_0-1\}$.   
\end{lemma}

The proof of Lemma~\ref{lem:chain-rule} in given in Section~\ref{sec:chain}.

\subsubsection{The guessing lemma}

We state the last lemma required for the proof of Theorem~\ref{thm:main-tech}. A similar lemma already appeared in~\cite{VV12}. Here we give a slightly more general version of the lemma stated in a form that can be directly used in the proof of the theorem. 

\begin{lemma}[Guessing lemma]\label{lem:guessing}
Let $\delta,\nu,\eta>0$. Suppose given six bipartite states $\rho_{\boxA\boxB}^{xy}$, where $x\in\{0,1,2\}$, $y\in\{0,1\}$, such that the following hold:
\begin{enumerate}
\item If $\rho_\boxA = (1/6)\sum_{xy}\Tr_\boxB(\rho_{\boxA\boxB}^{xy})$ and $\rho_\boxB = (1/6)\sum_{xy}\Tr_\boxA(\rho_{\boxA\boxB}^{xy})$,
\beq\label{eq:guess-0}
 \frac{1}{6}\sum_{x,y}\big\| \rho_\boxA - \rho_\boxA^{xy}  \big\|_1\leq \nu\qquad\text{and}\qquad\frac{1}{6}\sum_{x,y}\big\| \rho_\boxB - \rho_\boxB^{xy}  \big\|_1\leq \nu,
\eeq
\item There exists observables $A_x = A_x^0-A_x^1$, $B_y=B_y^0-B_y^1$ on $\boxA,\boxB$ respectively that satisfy 
\begin{align*}
\frac{1}{4}&\Big(\Tr\big( (A_0 \otimes B_0) \rho_{\boxA\boxB}^{00}\big) + \Tr\big( (A_0 \otimes B_1) \rho_{\boxA\boxB}^{01}\big) \\
&\qquad+ \Tr\big( (A_1 \otimes B_0) \rho_{\boxA\boxB}^{10}\big) - \Tr\big( (A_1 \otimes B_1) \rho_{\boxA\boxB}^{11}\big) \Big)\geq \frac{\sqrt{2}}{2} - \eta,
\end{align*}
\item Bob's measurement $B_1$ produces outcome $b_1\in\{0,1\}$ with probability $1-\delta$, when performed on his share of $\rho_{\boxA\boxB}^{21}$: 
$$\Tr( (\Id\otimes B_1^{b_1}) \rho_{\boxA\boxB}^{21} ) \geq 1-\delta.$$
\end{enumerate}
Then the condition 
$$\delta \, \geq \, \Big(\frac{\sqrt{2}-1}{2}- \eta\Big) - 75\nu$$
 must hold. 
\end{lemma}

\begin{proof}
For every $(a,b,x,y)\in\{0,1\}^2 \times \{0,1,2\}\times\{0,1\}$ let $p(a,b|x,y):= \Tr( (A_x^a\otimes B_y^b) \rho_{\boxA\boxB}^{xy})$. Condition~\eqref{eq:guess-0} implies that the distribution $p$ is approximately no-signalling, in the following sense: on average over the choice of a uniformly random pair $(x,y)$, the statistical distance 
\begin{align*}
\frac{1}{6}\sum_{x,y}\,\sum_{a}\,\Big| \sum_b \,p(a,b|x,y) - \frac{1}{2}\sum_{y'}\Big(\sum_b\, p(a,b|x,y') \Big)\Big| &\leq \frac{1}{6}\sum_{x,y}\,\sum_a \,\big| \Tr\big( (A_x^a \otimes \Id)(\rho_{\boxA\boxB}^{xy} - \rho_{\boxA\boxB}^{x})\big)\big|\\
&\leq \frac{1}{6}\sum_{x,y}\,\big\|\rho_{\boxA\boxB}^{xy} - \rho_{\boxA\boxB}^{x}\big\|_1\\
&\leq 2\nu,
\end{align*}
and a similar bound holds for the marginals on $\mathcal{B}$. Lemma~9.5 in~\cite{Hol09} implies that there exists a distribution $q(a,b|x,y)$ such that $q$ is (perfectly) no-signalling, and moreover, on average over $(x,y)$ the statistical distance $\|p(\cdot,\cdot|x,y)-q(\cdot,\cdot|x,y)\|_1\leq 10\nu$. In particular, the second assumption in the lemma implies that the distribution $q$ must violate the $\CHSH$ inequality by at least $\sqrt{2}/2-\eta-15\nu$, and the third assumption implies that $\sum_a q(a,1|2,1) \geq 1-\delta - 60\nu$. Applying the bound~(A.11) derived in the supplementary information to~\cite{Pironio} with $I/4 = \sqrt{2}/2 - \eta-15\nu$ we obtain the inequality claimed in the lemma. 
\end{proof}

\subsubsection{Proof of Theorem~\ref{thm:main-tech}}

We give the proof of Theorem~\ref{thm:main-tech}, assuming the lemmas stated in the three previous subsections.

\begin{proof}[Proof of Theorem~\ref{thm:main-tech}]
Let $(X,Y,A,B)$ be random variables describing Alice and Bob's choice of inputs to $\boxA$ and $\boxB$ respectively, and the outputs obtained, in an execution of Protocol~A. Let $E = E(\hat{\Adv})$ be the random variable that describes the outcome of the measurement on $\mathcal{E}$ described in Lemma~\ref{lem:strong-adv}, when the advice bits $\hat{\Adv}$ are selected uniformly at random (independently from $A$ and $B$). Denote by $\Adv = f_\Adv(B_\ck)A_\bl B_\bl$ the ``correct'' advice bits. 

The proof proceeds by contradiction. Assume that there existed a pair of devices $(\boxA,\boxB)$ such that 
\beq\label{eq:assumption}
\Pr\big(\CHSH_{\boxA \boxB}(\bl,\eta)\big)\,\geq\,\eps,\quad \Hmin^\eps(B_{\ck}|XYA_\bl B_\bl \mathcal{E}) \,<\, \kappa |\ck|,
\eeq 
where $\eps,\eta,\kappa$ are as in the statement of the theorem. 
Denote $\GUESS_{\boxB\boxE}(\hat{\Adv})$ the event that $E = B_\ck$. Using Lemma~\ref{lem:strong-adv}, we deduce from~\eqref{eq:assumption} that the following must hold: 
\begin{align}
\Pr\big(& \CHSH_{\boxA\boxB}(\bl,\eta)\wedge \GUESS_{\boxB\boxE}(\hat{\Adv}) | \hat{\Adv}=\Adv\big) \notag\\
&= \Pr\big( \GUESS_{\boxB\boxE}(\hat{\Adv}) | \CHSH_{\boxA\boxB}(\bl,\eta),\hat{\Adv}=\Adv\big)\notag\\
&\qquad\qquad\cdot \Pr\big(\CHSH_{\boxA\boxB}(\bl,\eta) | \hat{\Adv}=\Adv\big)\notag\\
&\geq\, C_E(\eps/m)^6 \cdot \eps,\label{eq:mainpf-1}
\end{align}
where $C_E$ is the constant from Lemma~\ref{lem:strong-adv}. Since the rounds $\bl$ are chosen uniformly at random, Claim~\ref{claim:high-chsh} below states that, for any $0\leq\beta\leq 1$:
\beq\label{eq:mainpf-2}
\Pr\big(\CHSH_{\boxA\boxB}((1+\beta)\eta)|\CHSH_{\boxA\boxB}(\bl,\eta)\big) \,\geq\, 1-e^{-2\beta^2\eta^2\gamma m},
\eeq
where $\gamma = |\bl|/m$. Choose $\beta = 1/3$, and let $\eta':=4\eta/3$. Provided $C_\gamma$ is chosen large enough, the choice of $\gamma$ made in the theorem is such that $\gamma \geq \log(2m^6/C_E\eps^7)/((2/9)\eta^2 m)$, so that $e^{-2\beta^2\eta^2\gamma m} \leq C_E\eps^7/(2m^6)$. Hence we obtain the following by combining~\eqref{eq:mainpf-1} and~\eqref{eq:mainpf-2}:
\begin{align}
\Pr\big(\CHSH_{\boxA\boxB}(\eta') \wedge \GUESS_{\boxB\boxE}(\hat{\Adv})|\,\hat{\Adv}=\Adv \big) 
\,\geq\, C_E(\eps^7/(2m^6))\, =:\,\eps'.\label{eq:assumption-2}
\end{align}
We may now apply Lemma~\ref{lem:chain-rule}. Let $\nu = C_\nu \sqrt{\log(1/\eps')/m}$, and $i_0\in [m]$ be the ``good'' round that is promised by the lemma. We proceed to show that the existence of such a round leads to a contradiction by appealing to the guessing lemma, Lemma~\ref{lem:guessing}.

Consider the following setup. Alice, Bob and Eve prepare their devices by selecting a random string of inputs $\hat{x},\hat{y}$ for Eve, except that $\hat{x}_{i_0}=2$ and $\hat{y}_{i_0} = 1$ always. Eve guesses the advice bits $\hat{\Adv}$ at random and makes a prediction $E=e$. Alice and Bob then use their devices up to round $i_0-1$ by choosing inputs $(x_{<i_0},y_{<i_0}) = (\hat{x}_{<i_0},\hat{y}_{<i_0})$. They verify that the resulting outputs $a_{<i_0},b_{<i_0}$ are such that 
$$(x_{<i_0},y_{<i_0},a_{<i_0},b_{<i_0},e_{<i_0})\in G_{i_0};$$
 if not they abort. Upon having succeeded in this conditioning they separate and play the guessing game. Alice holds system $\boxA$, while Bob holds system $\boxB$.

Lemma~\ref{lem:chain-rule} shows that all conditions in Lemma~\ref{lem:guessing} are satisfied: as a result, it must be that
$$ 12\ln(2)\alpha+\nu \geq \Big(\frac{\sqrt{2}-1}{2} - 6\eta'- 2\nu\Big) - 75\nu.$$
By definition, provided the constant $C_\nu$ is large enough we have $\alpha \leq \kappa/6 + 2 \gamma + \nu$, where we used that $|\ck|\leq m/6+10\sqrt{m}=m/6+O(\sqrt{\ln(1/\eps)})$, as enforced in the protocol, and $\eta'= 4/3\eta$. 
 Re-arranging terms and using the definition of $\nu$ and $\gamma$ we obtain the condition
$$ \kappa \,>\, \frac{\sqrt{2}-1}{4\ln(2)} - \frac{4}{\ln(2)} \eta - O\Big(\frac{\log(1/\eps)}{\eta^2 m}\Big),$$
which, given the choice of $\kappa$ made in the theorem, is a contradiction provided $C_\eps$ is chosen small enough. 
\end{proof}

\begin{claim}\label{claim:high-chsh}
Let $\eta,\gamma>0$. The following holds for any $0\leq\beta\leq 1$:
$$\Pr_{S}\big(\CHSH((1+\beta)\eta)|\CHSH(S,\eta)\big) \,\geq\, 1-e^{-2\beta^2\eta^2\gamma m},$$
where the probability is taken over the choice of a random subset $S\subseteq [m]$ of size $|S| = \gamma m$.
\end{claim}

\begin{proof}
Consider a given run of the protocol. Suppose that the fraction of rounds in which the $\CHSH$ condition is not satisfied is at least $(1-\opt)+(1+\beta)\eta$. By a standard Chernoff bound, a randomly chosen set $S\subseteq [m]$ will of size $\gamma m$ will have at least $((1-\opt)+\eta)\gamma m$ of its rounds with inputs corresponding to the $\CHSH$ condition being violated, except with probability at most  $e^{-2\beta^2\eta^2 \gamma m}$.
\end{proof}

\section{Proof of Lemma~\ref{lem:chain-rule}}\label{sec:chain}

This section is devoted to the proof of Lemma~\ref{lem:chain-rule}.  Let $D$ be the event $\CHSH_{\boxA\boxB}(\eta) \wedge \GUESS_{\boxB\boxE}$: the main assumption of the lemma states that $\Pr(D|\Adv=\hat{\Adv})\geq \eps$. We first prove two preliminary claims which establish that, provided $\eps$ is not too small, conditioning on $D$ does not affect either the distribution of inputs $(X_i,Y_i)$ or the reduced density matrices of the inner state of each device's system in most rounds $i$ by too much.

\begin{claim}\label{claim:prod_inputs}
Suppose that, in Protocol~B, Alice and Bob choose inputs $(X,Y)\in\{0,1,2\}^m \times \{0,1\}^m$ uniformly at random, obtaining outcomes $A,B\in\{0,1\}^m$. Suppose that $\mathcal{E}$ is measured using Eve's guessing measurement (as described in Lemma~\ref{lem:strong-adv}) with inputs $(\hat{X},\hat{Y})=(X,Y)$ and advice bits $\hat{\Adv}=\Adv$, resulting in an outcome $E \in \{0,1\}^{|\ck|}$. Let $P_{X_iY_i}$ be the marginal distribution of the inputs in the $i$-th round, conditioned on $(X_{<i},Y_{<i},A_{<i},B_{<i},E_{<i}) = (x_{<i},y_{<i},a_{<i},b_{<i},e_{<i})\in D_{<i}$, the projection of $D$ on the first $(i-1)$ coordinates. Then the following bound holds on expectation over $(x_{<i},y_{<i},a_{<i},b_{<i},e_{<i})$:
$$ \frac{1}{m} \sum_i \big\|P_{X_iY_i} - U_{3\times 2} \big\|_1 \,\leq\,  \sqrt{\frac{\log(1/\eps)}{2m}},$$
where $U_{3\times 2}$ is the uniform distribution on $\{0,1,2\}\times \{0,1\}$. 
\end{claim}

\begin{proof}
The Shannon entropy $H(X,Y) = \log(6)m$, and conditioned on $D$, $H(X,Y|D) \geq \log(6)\,m - \log(1/\eps)$. Applying the chain rule,
$$ \frac{1}{m} \sum_i H(X_i,Y_i | X_{<i},Y_{<i},D_{<i})\,\geq\, \log(6) - \frac{\log(1/\eps)}{m}.$$
Using the classical Pinsker's inequality as $\|P_{X_iY_i}-U_{3\times 2}\|_1 \leq \sqrt{(\log(6)-H(X_i,Y_i))/2}$ and Jensen's inequality we get
$$
\frac{1}{m}\sum_i \big\|P_{X_iY_i}-U_{3\times 2}\big\|_1\,\leq\, \sqrt{  \frac{\log(1/\eps)}{2m}},
$$
proving the claim. 
\end{proof}

The fact that $D$ depends both on the choice of inputs $(X,Y)$ and on the adversary's measurement outcome implies that conditioning on $D$ could not only bias the distribution of $(X,Y)$ but also introduce correlations between $(X,Y)$ and the reduced state $\rho_{\mathcal{A}\mathcal{B}}$ of the devices. The following claim shows that, if $D$ is an event with large enough probability, the correlations introduced by this conditioning do not affect the reduced state on either $\mathcal{A}$ or $\mathcal{B}$ by too much, for most rounds $i$. 

\begin{claim}\label{claim:product}
Consider the same situation as described in Claim~\ref{claim:prod_inputs}. Let $\rho_{\boxA_i X_iY_i}$ denote the reduced density of the joint state of systems $\boxA$ (in round $i$) and $X_i,Y_i$, conditioned on $(X_{<i},Y_{<i},A_{<i},B_{<i},E_{<i}) = (x_{<i},y_{<i},a_{<i},b_{<i},e_{<i})\in D_{<i}$. Then the following holds on expectation over $(x_{<i},y_{<i},a_{<i},b_{<i},e_{<i})$:
\beq\label{eq:product-0}
 \frac{1}{m} \sum_i \,\Big\| \rho_{\boxA_i X_iY_i} - \rho_{\boxA_i}\otimes \Big(\frac{1}{6}\sum_{x,y} \ket{x,y}\bra{x,y}\Big) \Big\|_1 \,\leq\, 4 \sqrt{\log(1/\eps)/m}.
\eeq
Moreover, the same bound holds when $\boxA_i$ is replaced by $\boxB_i$. 
\end{claim}

\begin{proof}
We use Claim~\ref{claim:mutualinfo}. Alice's sequential measurements are taken to be the ones performed on $\boxA$, while Bob's measurement is the combination of the measurements on $\boxB$, together with Eve's measurement, on inputs $X,Y$ and advice bits $\hat{\Adv}=\Adv$ obtained from $B$. We set $\mathbf{X}$ in the claim to be $XY$ here, and the outcomes $\mathbf{B}$ in the claim to $BE$ here. Together with the assumption $\Pr(D|\hat{\Adv}=\Adv)\geq \eps$, the claim shows that
$$\frac{1}{m} \sum_i \,I\big(\boxA_i;X_iY_i|D_{<i}\big)_{\rho_{\boxA_i X_iY_i}} \,\leq\, \frac{\log(1/\eps)}{m}.$$
Using Pinsker's inequality~\eqref{eq:pinsker} together with Jensen's inequality,  
$$\frac{1}{m}\sum_i \Big\|\rho_{\boxA_i X_i Y_i} -  \rho_{\boxA_i}\otimes \Big(\frac{1}{6}\sum_{xy} \ket{x,y}\bra{x,y}\Big)  \Big\|_1 \,\leq\, 4\sqrt{\log(1/\eps)/m},$$ 
where we used Claim~\ref{claim:prod_inputs} to show that the marginal distribution of $(X_i,Y_i)$ is close to uniform on $\{0,1,2\}\times\{0,1\}$, even conditioned on $D_{<i}$. 
\end{proof}

The following claim replaces the event that the $\CHSH$ condition is satisfied in a large fraction of rounds by the event that their exists many rounds in which the $\CHSH$ condition is \emph{likely} to be satisfied (when evaluated on the state of the devices in that round). 

\begin{claim}\label{claim:unif}
There exists a set $T\subseteq [m]$ such that $|T|\geq 2m/3$, and a subset $D'\subseteq D$ such that $\Pr(D'|D)\geq 1/2$ and for every $i\in T$, conditioned on $\hat{\Adv}=\Adv$ and on inputs and outputs to the devices in rounds prior to $i$ being in $D'$, the condition $\VIOL_{\boxA\boxB}(i)\leq 3\eta + 6\sqrt{\ln(1/\eps)/m}$ holds. 
\end{claim}

\begin{proof}
Let $Z_i\in\{0,1\}$ be $1$ if and only if the $\CHSH$ condition is not satisfied in round $i$. By definition, $\textrm{E}[Z_i] = (1-\opt) + \VIOL_{\boxA\boxB}(i)$. Let $W_i = \textrm{E}[Z_i] - Z_i$ and $W_{\leq i} = W_1+\cdots+W_i$. $(W_{\leq i})_i$ is a Martingale, and by Azuma's inequality, for any $\beta >0 $
\begin{align}
 \Pr\Big( \frac{1}{m} \sum_i \VIOL_{\boxA\boxB}(i) + (1-\opt) > \frac{1}{m}\sum_i Z_i + \beta\Big) 
&= \Pr\Big( \frac{1}{m}\sum_i W_i > \beta \Big) \notag\\
& \leq e^{-\beta^2 m/2}.\notag
\end{align}
Since the string $\hat{\Adv}$ is chosen by the adversary uniformly at random, we may further condition the equations above on $\hat{\Adv}=\Adv$ without affecting their validity. 
Note that the event  $\CHSH_{\boxA\boxB}(\eta)$ is equivalent to $\frac{1}{m}\sum_i Z_i \leq (1-\opt)+\eta$. Choosing $\beta = \sqrt{2\ln(2/\eps)/m}$, so that $e^{-\beta^2 m/2} < \eps/4$, and using the assumption $\Pr(D|\hat{\Adv}=\Adv)\geq \eps$ to further condition on $D=\CHSH_{\boxA\boxB}(\eta)\wedge\GUESS_{\boxB\boxE}$ we get
$$
\Pr\Big( \frac{1}{m}\sum_i \VIOL_{\boxA\boxB}(i) > \eta + \beta  \big| D, \hat{\Adv}=\Adv \Big) \,\leq\, 1/2.
$$
The quantity $\VIOL_{\boxA\boxB}(i)$ is a nonnegative number which only depends on the state of the devices in round $i$, itself only depending on the string of inputs and outputs observed thus far. Applying Markov's inequality, the condition above implies that there is a set $T\subseteq[m]$ of size $|T|\geq 2m/3$ and a subset $D'\subseteq D$ of size $\Pr(D'|D)\geq 1/2$ such that for every $i\in T$ it holds that $\VIOL_{\boxA\boxB}(i) \leq 3( \eta + \beta)$, provided previous inputs and outputs of the devices were in $D'$. 
\end{proof}

\begin{proof}[Proof of Lemma~\ref{lem:chain-rule}]
Let $D'$ be the set from Claim~\ref{claim:unif}. Consider the state of the devices $\boxA$ and $\boxB$ in an arbitrary round $i$ of the protocol. By applying Markov's inequality to the bound~\eqref{eq:product-0} from Claim~\ref{claim:product}, we obtain a set $|T'|\subseteq[m]$ of size $|T'|\geq 11m/12$ and a subset $D''\subseteq D'$ satisfying $\Pr(D''|D')\geq 1/2$ such that, for every $i\in T'$, conditioned on $\hat{\Adv}=\Adv$ and $(X_{<i},Y_{<i},A_{<i},B_{<i},E_{<i}) = (x_{<i},y_{<i},a_{<i},b_{<i},e_{<i})\in D''_{<i}$, both bounds
$$
\Big\| \rho_{\boxA_i X_iY_i} - \rho_{\boxA_i}\otimes \Big(\frac{1}{6}\sum_{x,y} \ket{x,y}\bra{x,y}\Big) \Big\|_1 \,\leq\, 200 \sqrt{\log(1/\eps)/m}
$$
and the analogous bound where $\mathcal{A}_i$ is replaced by $\mathcal{B}_i$ hold. Letting $T''=T'\cap T$, where $T$ is the set from Claim~\ref{claim:unif}, both the bound above and the condition $\VIOL_{\boxA\boxB}(i)\leq 3\eta + 6\sqrt{\ln(1/\eps)/m}$ hold simultaneously in the rounds from $T''$ (conditioned on previous inputs and outputs being in $D''$). Furthermore, note that whether both conditions are satisfied or not only depends on the (post-selected) state of the protocol in round $i$, itself only depending on subsequent choices of inputs and outputs in the protocol to the extent that the condition $\hat{\Adv}=\Adv$ is satisfied. Hence as long as the advice bits $\hat{\Adv}$ that Eve uses to select the measurement on her system have a positive probability of being the correct advice bits, given the data generated up to round $i-1$, both bounds must hold verbatim. As a consequence, for any fixed $(x,y,a,b,e)\in D''_{<i}$ there exists a string $(\hat{x}_{>i},\hat{y}_{>i},\hat{a}_{>i},\hat{b}_{>i})$ from which advice bits $\hat{\Adv}_{>i}$ can be computed such that if Eve makes the corresponding measurement, and obtains outputs that match $e_{<i}$, the bounds will hold irrespective of what might happen if the protocol was to be run for rounds after $i$. Thus conditions~\eqref{eq:chain-0} and~\eqref{eq:chain-1} in the lemma hold for any round $i\in T''$. 

It remains to show that condition~\eqref{eq:chain-1} holds simultaneously in some round $i_0$. Since by construction $\Pr(D''|\hat{\Adv}=\Adv) \geq \eps/4$, multiplying by $\Pr(\hat{\Adv}=\Adv)=2^{-\alpha m}$, applying Baye's rule, and using the definition of $D =  \CHSH_{\boxA\boxB}(\eta)\wedge \GUESS_{\boxB\boxE}$, we get
$$ \prod_{i=1}^m \Pr\big(\GUESS_{\boxB\boxE}(i)|D''_{<i}\big) \,\geq\,(\eps/4)\,2^{-\alpha m}.$$
Taking logarithms and applying Markov's inequality, there is a subset $S\subseteq[m]$ of size $|S|\geq m/2$ such that for every $i\in S$, 
$$-\ln \Pr\big(\GUESS_{\boxB\boxE}(i)|D''_{<i}\big) \leq 2(\ln(2)\alpha + \ln(4/\eps)/m),$$
implying that, for all $i\in S$,
\beq\label{eq:cr-1}
\Pr\big(\GUESS_{\boxB\boxE}(i)|D''_{<i})\geq 1-2\ln(2)\alpha-2\ln(4/\eps)/m.
\eeq
Let $i_0$ be any round in $T''\cap S$. To obtain~\eqref{eq:chain-2} we need to further condition~\eqref{eq:cr-1} on inputs in round $i_0$ to be the pair $(2,1)$, which using Claim~\ref{claim:prod_inputs} happens with probability $1/6\pm O(\sqrt{\ln(1/\eps)/m})$.  Choosing $C_\nu$ in the lemma to be a large enough constant, all three conditions are satisfied. 
\end{proof}

\section{The quantum reconstruction paradigm}\label{sec:adv}

In this section we prove a general lemma, Lemma~\ref{lem:ext_adv} in Section~\ref{sec:rec-lem} below, from which Lemma~\ref{lem:strong-adv} is deduced in Section~\ref{sec:strong-adv}. We start with some useful preliminary definitions and known results. 

\subsection{Combinatorial preliminaries}

We first define extractors. 

\begin{definition}\label{def:extractor}
  A function $Ext: \{0,1\}^n \times \{0,1\}^d \to \{0,1\}^m$ is a
  \emph{quantum-proof} (or simply \emph{quantum})
  \emph{$(k,\eps)$-strong extractor} if for all
  states $\rho_{XE}$ classical on $X$ with $\Hmin(X|E) \geq k$,
  and for a uniform seed $Y\in\{0,1\}^d$, we have 
	$$\frac{1}{2} \big\|
    \rho_{Ext(X,Y)YE} - \rho_{U_m} \tensor \rho_Y \tensor \rho_E\big\|_1
  \leq \eps, $$
	where $\rho_{U_m}$ is the fully mixed state on a system of dimension $2^m$.
\end{definition}

We will use list-decodable codes.

\begin{definition}
  A code $C : \{0,1\}^n \to \{0,1\}^{\bar{n}}$ is said to be
  $(\eps,L)$-list-decodable if every Hamming ball of relative
  radius $1/2 - \eps$ in $\{0,1\}^{\bar{n}}$ contains at most
  $L$ codewords.
\end{definition}

There exist list-decodable codes with the following parameters.

\begin{lemma} \label{lem:ecc} For every $n \in \N$ and $\delta > 0$
  there is a code $C_{n,\delta} : \{0,1\}^n \to \{0,1\}^{\bar{n}}$,
  which is $(\delta,1/\delta^2)$-list-decodable, with $\bar{n} = \poly(n,1/\delta)$. Furthermore, $C_{n,\delta}$ can be evaluated in time
  $\poly(n,1/\delta)$ and $\bar{n}$ can be assumed to be a power of
  $2$.
\end{lemma}

For example, Guruswami et al.~\cite{GHSZ02} combine a Reed-Solomon
code with a Ha\-da\-mard code, obtaining such a list-decodable code with
$\bar{n} = O(n/\delta^4)$.

We will also use the notion of weak design, as defined in~\cite{RRV02}.

\begin{definition}\label{def:weakdesign}
 A family of sets $S_1,\cdots,S_m \subset [d]$
  is a \emph{weak $(t,r,m,d)$-design} if
  \begin{enumerate}
    \item For all $i$, $|S_i| = t$.
    \item For all $i$,  $\sum_{j = 1}^{i-1} 2^{|S_j \cap S_i|} \leq rm$.
  \end{enumerate}  
\end{definition}

There exists designs with the following parameters. 

\begin{lemma}[\protect{\cite[Lemma 17]{RRV02}}] \label{lem:optimalweakdesign}
 For every $t,m \in \N$ there
  exists a weak $(t,1,m,d)$-design $S_1,\dotsc,S_m \subset [d]$ such that
  $d = t \left\lceil \frac{t}{\ln 2} \right\rceil \left\lceil\log 4m
  \right\rceil = O(t^2 \log m)$. Moreover, such a design can be found
  in time $\poly(m,d)$ and space $\poly(m)$.
\end{lemma}

Finally, we describe Trevisan's extractor construction. 

\begin{definition}\label{def:genericscheme} For a
  one-bit extractor $C : \{0,1\}^n \times \{0,1\}^t \to \{0,1\}$,
  and for a weak $(t,r,m,d)$-design $S_1,\cdots,S_m \subset [d]$, we define the
  $m$-bit extractor $Ext_C : \{0,1\}^n \times \{0,1\}^d \to
  \{0,1\}^m$ as 
	$$ Ext_C(x,y) \,:=\, C(x,y_{S_1}),\ldots,C(x,y_{S_m}).$$
	\end{definition}

\subsection{The reconstruction lemma}\label{sec:rec-lem}

The following lemma is implicit in the proof of security of Trevisan's extractor construction paradigm against quantum adversaries given in~\cite{DVPR11}. A similar lemma also appeared in~\cite[Lemma~13]{VV12}, where the code $C$ was specialized to the $t$-XOR code. For completeness, we state and sketch the proof of a more general variant of that lemma. 

\begin{lemma}\label{lem:ext_adv} Let $n,m,r,t,L$ be integers and $\eps>0$. Let $C:\{0,1\}^n \to \{0,1\}^{\bar{n}}$ be a $(\eps^2/(8m^2),L)$-list-decodable code, where $\bar{n}=2^t$. Let $Ext_C$ be the extractor obtained by combining $C$ with a $(t,r,m,d)$ design as in Definition~\ref{def:genericscheme}.

Let $\rho_{XE}$ be a state such that $X$ is a random variable distributed over $n$-bit strings. Let $U_m$ be uniformly distributed over $m$-bit strings, and suppose that 
\beq\label{eq:ass-eve}
\|\rho_{Ext_C(X,Y)YE} - \rho_{U_m}\otimes \rho_Y \otimes \rho_E \big\|_{1} \, > \, \eps,
\eeq
where $Y$ is uniformly distributed over $\{0,1\}^d$. Then there exists fixed strings $y_1,\ldots,y_{rm} \in\{0,1\}^t$ such that, given the $\{(y_i,C(X)_{y_i})\}$ as advice, with probability at least $\eps^2/(8m^2)$ over the choice of $x\sim p_X$ and her own randomness an ``adversary'' Eve holding system $E$ can produce a string $z$ such that $d_H(z , C(x)) \leq 1/2 - \eps^2/(8m^2)$. In particular, Eve can recover $L$ strings $\tilde{x}_i \in \{0,1\}^n$ such that there exits $i$, $\tilde{x}_i = x$.  
\end{lemma}

\begin{proof}
Proposition~4.4 from~\cite{DVPR11} shows that a standard hybrid argument, together with properties of Trevisan's extractor (specifically the use of the seed through combinatorial designs), can be used to show the following claim.

\begin{claim}\label{claim:ext_adv-1} Assume~\eqref{eq:ass-eve} holds. Then there exists strings $y_1,\ldots, y_{rm} \in \{0,1\}^t$, and for every $y\in\{0,1\}^t$ a binary measurement, depending on the $\{(y_i,C(X)_{y_i})\}$, on $E$ that outputs $C(X,y)$ with probability at least $1/2+\eps/m$ on average over $y$. Formally, 
\beq\label{eq:eve-1}
\big\| \rho_{C_t(X)_Y Y V E} - \rho_{U_1}\otimes \rho_{Y} \otimes \rho_{VE} \big\|_{1} \,>\, \frac{\eps}{m},
\eeq
where $Y$ is a random variable uniformly distributed over $\{0,1\}^t$ and $V$ is a classical register containing the $\{(y_i,C(X)_{y_i})\}$. 
\end{claim}

The next step is to argue that Eq.~\eqref{eq:eve-1} implies that an adversary given access to $E'=VE$ can predict not only a random bit of $C(X)$, but a string $Z$ of length $m$ such that $Z$ agrees with $C(X)$ in a significant fraction of positions. This follows from an argument given in~\cite{KT07}, and the following claim is proved exactly as~\cite[Claim~15]{VV12}. 

\begin{claim}\label{claim:ext_adv-2} Suppose~\eqref{eq:eve-1} holds. Then there exists a measurement $\mathcal{F}$, with outcomes in $\{0,1\}^n$, such that 
\beq\label{eq:eve-5}
\Pr_{x\sim p_X,\,y\sim U_{t}}\big(\, C(x)_y \,= \,C(\mathcal{F}(VE))_y\,\big) \geq \frac{1}{2} + \frac{\eps^2}{4m^2}\, ,
\eeq
where $\mathcal{F}(VE)$ denotes the outcome of $\mathcal{F}$ when performed on the state $\rho_{VE}$.
\end{claim}

To conclude the argument, we use the error-correction properties of $C$ to argue that Eve can decode her string $C(\mathcal{F}(VE))$ into an educated guess of $x$. Claim~\ref{claim:ext_adv-2} shows that, on expectation over $x$, Eve's string is at Hamming distance $1/2-\eps^2/(4m^2)$ from the encoding of $x$. In particular, the distance will be at most $1/2-\eps^2/(8m^2)$ for a fraction at least $\eps^2/(8m^2)$ of $x\sim p_X$. 
Since, by assumption, $C$ is $(\eps^2/(8m^2),L)$-list-decodable, for those $x$ Eve can narrow down the possibilities to at most $L$ distinct values. 
\end{proof}

\subsection{Proof of Lemma~\ref{lem:strong-adv}}\label{sec:strong-adv}

The proof of Lemma~\ref{lem:strong-adv} follows immediately from Lemma~\ref{lem:ext_adv} and an appropriate choice of parameters. Let $E$ denote the system made of the combination of $XYA_\bl B_\bl \mathcal{E}$, and let $n=|\ck|$. The assumption of the lemma is that $\Hmin^\eps(B_\ck|E) < \kappa n$. Let $m = \kappa n + 1$. Let $C = C_{n,\delta}$, where $\delta = \eps^2/(32 m^2)$, be a $(\delta,1/\delta^2)$ list-decodable code, as promised by Lemma~\ref{lem:ecc}. Let $Ext_C$ be constructed from $C$ and a $(t,1,m,d)$ design, where $t = \log \bar{n}$ and $d = O(t^2\log m)$, as promised by Lemma~\ref{lem:optimalweakdesign}. 

It follows from the data processing inequality (see e.g.~\cite[Lemma~V.1~(ii)]{KR11}), our assumed upper bound on $\Hmin^\eps(B_\ck|E)$, and our choice of $m$ that Eq.~\eqref{eq:ass-eve} holds with $(\eps/2)$ in place of $\eps$. Thinking of Eve as simply outputting one of her $L$ guesses $\tilde{x}_i$ chosen at random, we obtain that Eve's guess will be successful with probability at least $\eps^2/(32L m^2)$. Overall, Eve needs $m$ bits of advice, given which she can predict $x$ with success probability $O(\eps^6/m^6)$, given  our choice of parameters.

\section{Additional lemmas}\label{sec:additional}

\begin{lemma}[Azuma-Hoeffding inequality]\label{lem:azuma}
Let $(X_k)$ be a martingale such that $|X_k-X_{k-1}|\leq c_k$ for all $k$. Then for all integers $m$ and all $t\geq 0$,
$$ \Pr\big( X_m - X_0\geq t \big)\,\leq\, e^{-t^2/(2\sum_k c_k^2)}.$$
\end{lemma}

\begin{lemma}\label{lem:conditioning} Let $\eps,\delta,\eta,\beta>0$ and $m$ an integer such that $e^{-2\beta^2 \delta m} < \eps/2$. Let $X$ be a random variable defined over $m$-bit strings. Suppose that $\Pr( \sum_i X_i \leq \eta m) \geq \eps$. Then there exists  a set $G\subseteq\{0,1\}^m$ such that $\Pr(G)\geq \eps/2$ and for all $x$ in $G$, for a fraction $\geq 1-\delta$ of indices $i\in [m]$,  
$$\Pr(X_i=0 | X_{<i} = x_{<i}) \geq 1-\eta-\beta.$$
As a consequence, for a fraction at least $1-2\delta$ of $i\in [m]$ there exists a set $G_i \subseteq G$ such that $\Pr(G_i|G)\geq 1/2$ and for every $x_{<i}\in G_{i}$,
$$\Pr(X_i=0 | X_{<i} = x_{<i} ) \geq 1-\eta-\beta.$$
\end{lemma}

\begin{proof}
For every $i\in [m]$ define 
$$B_i = \big\{(x_1,\ldots,x_{i-1},\ldots,x_m)|\, \Pr(X_{i}=1 | X_{<i} = x_{<i}) \geq \eta+\beta \big\},$$
let  
$$B \,=\, \Big\{ x\big|\, \sum_{i:x\in B_i}1 \geq \delta m\Big\},$$
and suppose towards a contradiction that $\Pr(B) \geq 1-\eps/2$. Let $\hat{B} = \{x\in B|\sum_i x_i \leq \eta m\}$. By definition, for every $x\in B$ and at least a $\delta$-fraction of indices $i$ it holds that  $\Pr(X_i=1 | X_{<i} = x_{<i}) \geq \eta+\beta$. Hence the probability that $x\in B$ has less than $\eta$ indices $j$ at which $x_j=1$ is at most $e^{-2\beta^2\delta m}$, i.e. $\Pr(\hat{B}|B) \leq e^{-2\beta^2\delta m}$. This shows that 
$$ \Pr\big(\sum_i X_i > \eta m \big) \,\geq\, \Pr(B)\big(1-\Pr(\hat{B}|B)\big) \,\geq\, (1-\eps/2)\big(1-e^{-2\beta^2\delta m}\big) \,>\, 1-\eps$$
given our assumption on $\eps,\delta,\eta,\beta$ and $m$; a contradiction. 

For the ``consequence'', for any $x\in G$ and $i\in [m]$ let $Y_{x,i}=1$ if and only if the condition
$$\Pr(X_i=0 | X_{<i} = x_{<i}) \geq 1-\eta-\beta$$
is satisfied. We have shown $\Es{x\in G, i\in [m]} \big[Y_{x,i}\big] \geq 1-\delta$. The result is then a consequence of Markov's inequality. 
\end{proof}

\begin{claim}\label{claim:mutualinfo} Let $\rho = \rho_{\boxA\boxB}$ be a bipartite state shared between Alice and Bob. Suppose Bob chooses $x\in \mathbf{X}^m$ according to distribution $(p_x)$, and applies a measurement with Krauss operators $\{N_x^b\}_{b\in\mathcal{B}^m}$ on $\boxB$. Alice \emph{sequentially} applies a measurement with Krauss operators $\{M_{x_i}^{a_i}\}_{a_i\in\mathbf{A}}$ on $\boxA$, for $i=1,\ldots,m$. Let $D\subseteq (\mathbf{X}\times \mathbf{A} \times \mathbf{B})^m$ be a set of probability $\Pr(D) = \eps$. For $i\in [m]$, let $\rho_i$ be the state of the system $\boxA\boxB X_i$ after $i-1$ measurements have been performed by Alice, conditioned on $(x_{<i},a_{<i},b_{<i})\in D_{<i}$:
$$ \rho_i \,\propto\, \sum_{(x,a,b): (x_{<i},a_{<i},b_{<i})\in D_{<i}} p_x\,\Big(\Big( \prod_{j<i}M_{x_j}^{a_j} \Big)\otimes N_x^b \Big)\,\rho\,\Big(\Big( \prod_{j<i} \big(M_{x_j}^{a_j}\big)^\dagger \Big)\otimes \big(N_x^b\big)^\dagger\Big),$$
and $\rho_i$ is normalized. Then the following bound holds:
$$ \sum_i \,I(\boxA:X_i | D_{<i})_{\rho_i} \,\leq\, \log(1/\eps).$$
\end{claim}

\begin{proof}
We prove the lemma using standard techniques from quantum information theory; specifically the proof of the Holevo-Schumacher-Westmoreland theorem~\cite{Holevo98,SchuW96}. We assume that the reader is familiar with the coding and decoding strategies employed in that result, and in particular the notion of typical subspace (see e.g.~[Chapters~14 and~19]\cite{Wilde11}, and more specifically the proof of Theorem~19.3.1). We prove the claim by describing an experiment by which Bob transmits $H(X)$ bits of information to Alice using only $H(X)+\log(1/\eps)-\sum_i I(\boxA:X_i)_{\rho_i}$ bits of communication from him to Alice. This implies the claimed inequality: if it did not hold Alice could guess Bob's $H(X)$ bits with success larger than $2^{-H(X)}$ simply by running the protocol by herself, and guessing Bob's messages.  

Suppose Alice and Bob share an infinite number of copies of $\rho$. For each $i\in [m]$, Alice and Bob also agree on a random code $\mathcal{C}_i\subseteq \mathcal{X}^K$, where $K$ is a large integer, such that $|\mathcal{C}_i| = 2^{K I(\boxA:X_i|D_{<i})_{\rho_i}}$.  By the properties of typical subspaces, with high probability over the choice of $\mathcal{C}_i$ the collection of states $\otimes_{j=1}^K\rho_i(x'_j)$ for $(x'_1,\ldots,x'_K)\in \mathcal{C}$, where $\rho_i(x'_j)$ is the reduced density of $\rho_i$ on $\mathcal{A}$ conditioned on $X_i = x'_j$, are almost perfectly distinguishable.\footnote{Precisely, there exists a distinguishing measurement whose success probability can be made arbitrarily close to $1$ by taking $K$ large enough.}

The experiment proceeds as follows. The copies of $\rho$ are grouped in groups of $K$. For each group, Bob selects a random $x=(x_i^j)_{1\leq i \leq m,1\leq j \leq K}\in (\mathcal{X}^m)^K$ and applies the measurements $\{N_{x^j}\}$ in the $j$-th copy of $\rho$ in that group, obtaining an outcome $b^j\in \mathbf{B}^m$. For each $i\in [m]$, Alice does the following, independently for each group. She guesses whether Bob's choice of $(x_i^1,\ldots,x_i^K)$ is in $\mathcal{C}_i$ (the probability with which she guesses this should be so is equal to the probability that $x_i \in \mathcal{C}_i$, i.e. $2^{K (I(\boxA:X_i|D_{<i})_{\rho_i}-H(X_i))}$). If so, she performs the decoding measurement to recover $x_i$. If not, she guesses $(x_i^1,\ldots,x_i^K)$ according to $p^{\times K}$. She then applies the measurements $\{M_{x_i^j}^{a_i^j}\}$ corresponding to the guessed $(x_i^j)$. At the end of the $m$ repetitions, Alice sends all her guesses, and her outcomes, to Bob. 

Finally, Bob finds the first group of $K$ states in which Alice's guesses were all correct, and $(x^j,a^j,b^j)\in D$ (for each $1\leq j \leq K$). In any group, the probability that this event happens is $2^{-K(H(X)-\sum_i I(\boxA:X_i|D_{<i})_{\rho_i})} \eps^K$. Moreover, note that Alice's probability of correctly guessing Bob's choice of $(x_i^j)$ is independent of $(x_i^j)$. Hence Bob can indicate to Alice the index of the first group of states on which she was correct by transmitting $O(K\log(1/\eps) + K(H(X)-\sum_i I(\boxA:X_i|D_{<i})_{\rho_i}))$ bits. Alice then knows all $KH(X)$ bits of information about Bob's choices of $x$ in the $m$ rounds on the group of $K$ states.  
\end{proof}

\bibliographystyle{alphaabbrvprelim}
\bibliography{randomness}
\end{document}